\documentclass[11pt,a4paper]{article}

\usepackage{setspace}   
\usepackage{enumitem}
\usepackage{amsmath, amssymb, amsthm, amsfonts, mathtools, bm}
\usepackage{lmodern}
\usepackage{microtype}

\usepackage[margin=1.00in]{geometry} 
\usepackage{cuted}           
\usepackage{IEEEtrantools}   

\allowdisplaybreaks
\usepackage[numbers,sort&compress]{natbib}
\usepackage{doi}
\usepackage{hyperref}
\hypersetup{
  colorlinks=true,
  linkcolor=red,   
  citecolor=green,
  urlcolor=black
}

\title{Two-Stage Asymmetric Tullock Contests with Cost Shifters and Endogenous Continuation Decision}

\author{Felix Reichel\thanks{Department of Economics, Johannes Kepler University Linz, Email: \url{mailto:felix.reichel@jku.at}}}
\date{September 30, 2025}

\theoremstyle{plain}
\newtheorem{theorem}{Theorem}
\newtheorem{proposition}{Proposition}

\newtheorem{corollary}{Corollary}
\newtheorem{assumption}{Assumption}
\theoremstyle{definition}
\newtheorem{definition}{Definition}
\theoremstyle{remark}
\newtheorem{remark}{Remark}

\newcommand{\RR}{\mathbb{R}}

\renewcommand*{\eqref}[1]{%
  \begingroup
    \hypersetup{linkcolor=green}
    \textup{(\ref{#1})}
  \endgroup
}

\begin{document}

\maketitle
\begin{abstract}
This paper introduces a contest-theoretic simplified model of triathlon as a sequential two-stage game. In Stage~1 (post-swim), participants decide whether to continue or withdraw from the contest, thereby generating an endogenous participation decision. In Stage~2 (bike--run), competition is represented as a Tullock contest in which swim drafting acts as a multiplicative shifter of quadratic effort costs. Closed-form equilibrium strategies are derived in the two-player case, and existence, uniqueness, and comparative statics are shown in the asymmetric $n$-player case. The continuation decision yields athlete-specific cutoff rules in swim drafting intensity and induces subgame-perfect equilibria (SPEs) with endogenous participation sets. The analysis relates swim drafting benefits, exposure, and group size to heterogeneous effective cost parameters and equilibrium efforts.
\end{abstract}
\vspace{1em}

\noindent\textbf{JEL codes:} C72, D74, D72, L83

\noindent\textbf{Keywords:} Contest theory; Tullock contest; triathlon; swim drafting; sequential games; asymmetric costs; subgame-perfect equilibrium

\newpage

\section{Contest Setup: Players, Timing, and Objects}\label{sec:setup}
There are $n\ge 2$ athletes indexed by $i\in N=\{1,\dots,n\}$. 
The contest is simplified to two stages.

\paragraph{Stage~1 (Post--Swim).} For each athlete $i$ we observe some state vector
\[
s_i=(t_i^S,r_i^S,D_i,G),
\]
where $t_i^S$ is swim time, $r_i^S$ is swim rank, $D_i\in[0,1]$ is the share of swim drafting, 
and $G$ is the full directed graph of stored swim drafting relations.  
The Stage~1 action set is $a_i\in\{C,T\}$: \emph{Continue Race (C)} or \emph{Withdraw from Race (T)}.

\paragraph{Stage~2 (Bike--Run Contest).}
Let $S\subseteq N$ be the subset of athletes who continue, with group size $m \leq |S|\ge 1$.  
Each $i\in S$ chooses a nonnegative effort $e_i$.  
The probability that $i$ wins is given by
\begin{equation}\label{eq:CSF}
p_i(e) = \frac{w_i e_i}{\sum_{j\in S} w_j e_j},
\end{equation}
where $w_i>0$ is a contest weight.  
Effort cost is quadratic with scaling parameter $k_i$, defined by
\begin{equation}\label{eq:cost}
C_i(e_i) = \tfrac12 k_i e_i^2,\qquad 
k_i := \frac{c_i}{\psi_i},
\end{equation}
where $c_i>0$ is a baseline cost parameter and $\psi_i$ is a swim drafting multiplier.  
The multiplier depends on the observed state,
\begin{equation}\label{eq:Psi}
\psi_i = \Psi(D_i,m,r_i^S,G)\in[\underline\psi,\overline\psi]\subset(0,\infty),
\end{equation}
with $\Psi$ increasing in swim drafting share $D_i$ and Lipschitz in $(D_i,m)$ where the reduced--drag model is $\psi_i=1/(1-\eta D_i)$ with $\eta\in(0,1)$.  
Furthermore let $\Delta_i>0$ denote the winner--loser prize differential relevant for $i$.

\paragraph{Payoffs.}
If athlete $i$ terminates in Stage~1, the outside payoff is
\begin{equation}\label{eq:outside}
U_i^T = \underbrace{-\alpha\, t_i^S}_{\text{time penalty}} 
        \;+\;\underbrace{-\beta\, r_i^S}_{\text{rank penalty}}
        \;+\;\underbrace{\vartheta_i}_{\text{additional training unit}},
\end{equation}
where $\alpha,\beta>0$ are penalty weights and $\vartheta_i\in\RR$ is an idiosyncratic term.  

If athlete $i$ continues to Stage~2, the expected utility is
\begin{equation}\label{eq:Ue}
U_i(e) = \underbrace{p_i(e)\,\Delta_i}_{\text{expected prize}} 
          \;+\;\underbrace{\big(-\tfrac12 k_i e_i^2\big)}_{\text{effort cost}} 
          \;+\;\underbrace{\text{(Terms independent of $e$)}}_{\text{irrelevant for effort choice}}.
\end{equation}

\begin{definition}[Strategy and Equilibrium]\label{def:equilibrium}
A (pure) strategy for athlete $i$ is $(a_i,\ e_i(\cdot))$, where $a_i\in\{C,T\}$ is the Stage~1 action and $e_i:\RR_+^m\to\RR_+$ is the Stage~2 effort choice.  
A \emph{subgame perfect equilibrium} (SPE) is a profile $(a_i^\ast,e_i^\ast)_{i\in N}$ such that:  
(i) given $S=\{i: a_i^\ast=C\}$, the effort profile $e^\ast$ is a Nash equilibrium (NE) of the Stage~2 contest;  
(ii) each $a_i^\ast$ maximizes total payoff given continuation values from (i).
\end{definition}

\section{Stage 2 Analysis: Existence, Uniqueness, and Closed Forms}\label{sec:stage2}
Stage 2 consists of the post--swim contest among the set of continuers $S$.

\subsection{Best responses}
Recall that athlete $i$'s probability of victory is
\[
p_i(e) = \frac{w_i e_i}{\sum_{j\in S} w_j e_j},
\]
where $e_i\ge 0$ is $i$'s bike--run effort and $w_i>0$ is a contest weight.  
The derivative $w.r.t.$ own effort is
\begin{equation}\label{eq:CSFder_r1}
\frac{\partial p_i}{\partial e_i} = \frac{p_i(1-p_i)}{e_i}, 
\qquad
\frac{\partial p_i}{\partial e_j} = -\frac{p_i p_j}{e_j},\ j\ne i.
\end{equation}

Dropping constants independent of $e_i$, athlete $i$ optimizes
\[
\max_{e_i\ge 0}\; p_i(e)\,\Delta_i - \tfrac12 k_i e_i^2,
\]
where $\Delta_i>0$ is $i$'s prize differential and $k_i=c_i/\psi_i$ is the effective cost parameter 
(baseline cost $c_i>0$, swim drafting multiplier $\psi_i>0$).  
The FOC for any interior optimum $e_i>0$ is
\begin{equation}\label{eq:FOC_r1}
0 = \frac{\partial U_i}{\partial e_i} 
   = \Delta_i\frac{p_i(1-p_i)}{e_i} - k_i e_i
   \quad\Longleftrightarrow\quad
   e_i^2 = \frac{\Delta_i}{k_i}\,p_i(1-p_i).
\end{equation}
Thus best responses map contest probabilities $p_i$ into equilibrium efforts via \eqref{eq:FOC_r1}.

\subsection{Two--player closed form ($m=2$)}
Suppose $S=\{i,j\}$.  
Define the \emph{effective advantage ratio}
\begin{equation}\label{eq:Rratio_r1}
R := \frac{\Delta_i \psi_i/c_i}{\Delta_j \psi_j/c_j} = \frac{\Delta_i/k_i}{\Delta_j/k_j}.
\end{equation}
Combining \eqref{eq:FOC_r1} with $p_i/(1-p_i)=(w_i e_i)/(w_j e_j)$ yields
\begin{equation}\label{eq:two_player_pi_r1}
\frac{p_i}{1-p_i} = R^{1/2}\,\frac{w_i}{w_j}.
\end{equation}
Hence $p_i^\ast=\rho/(1+\rho)$ with $\rho=R^{1/2}\,w_i/w_j$, and from \eqref{eq:FOC_r1}
\begin{align}\label{eq:two_player_e_r1}
e_i^\ast &= \sqrt{\tfrac{\Delta_i \psi_i}{c_i}}\;\frac{\rho^{1/2}}{1+\rho}, \\
e_j^\ast &= \sqrt{\tfrac{\Delta_j \psi_j}{c_j}}\;\frac{\rho^{1/2}}{1+\rho}. \nonumber
\end{align}
These closed forms show that equilibrium effort rises in $\Delta_i$ (prize) and in $\psi_i$ (swim drafting multiplier), and falls in $c_i$ (cost).

\subsection{Symmetric $m$--player benchmark}
If $w_i\equiv 1$, $c_i\equiv c$, $\psi_i\equiv \psi$, and $\Delta_i\equiv \Delta$ for all $i\in S$, then $p_i^\ast=1/m$.  
From \eqref{eq:FOC_r1},
\begin{equation}\label{eq:symm_e_r1}
(e^\ast)^2 = \frac{\Delta \psi}{c}\,\frac{m-1}{m^2},
\qquad
e^\ast = \sqrt{\tfrac{\Delta \psi}{c}}\;\sqrt{\frac{m-1}{m^2}}.
\end{equation}
Effort decreases with group size $m$, increases with prize $\Delta$ and swim drafting multiplier $\psi$, and decreases with cost parameter $c$.

\subsection{Asymmetric $m$--player case}
For general heterogeneity, let $E=\sum_{j\in S} e_j$ denote total effort.  
Since $p_i=e_i/E$, condition \eqref{eq:FOC_r1} implies
\begin{equation}\label{eq:EiFormula_r1}
e_i = \frac{\Delta_i\, E}{c_i E^2/\psi_i + \Delta_i},\qquad i\in S.
\end{equation}
Summing \eqref{eq:EiFormula_r1} over all $i$ gives the scalar equation
\begin{equation}\label{eq:E_fixed_point_r1}
1 = \sum_{i\in S}\frac{\Delta_i}{c_i E^2/\psi_i + \Delta_i}
   \quad\Longleftrightarrow\quad
   g(E) := \sum_{i\in S}\frac{\Delta_i}{c_i E^2/\psi_i + \Delta_i}-1=0.
\end{equation}

\begin{proposition}[Existence and uniqueness]\label{prop:exist_uni_r1}
Assume $c_i>0$, $\psi_i\in[\underline\psi,\overline\psi]$, and $\Delta_i>0$.  
Then $g(\cdot)$ is continuous and strictly decreasing on $E\ge 0$, with $g(0)=m-1>0$ and $\lim_{E\to\infty}g(E)=-1$.  
Hence there exists a unique $E^\ast>0$ solving \eqref{eq:E_fixed_point_r1}.  
The implied profile \eqref{eq:EiFormula_r1} is the unique interior NE of Stage~2.
\end{proposition}

\begin{proof}
Each term $h_i(E)=\Delta_i/(c_iE^2/\psi_i+\Delta_i)$ is strictly decreasing in $E\ge 0$, with $h_i(0)=1$ and $\lim_{E\to\infty}h_i(E)=0$.  
Thus $g(E)=\sum_i h_i(E)-1$ is strictly decreasing, crossing zero exactly once.  
Strict concavity of individual payoffs in $e_i$ (second derivative $-2\Delta_i/E^2-k_i<0$) guarantees a unique interior NE.
\end{proof}

\begin{corollary}[Probabilities and efforts]\label{cor:pi_e_r1}
At the unique equilibrium,
\[
p_i^\ast = \frac{\Delta_i}{c_i (E^\ast)^2/\psi_i + \Delta_i},\qquad
e_i^\ast = p_i^\ast E^\ast.
\]
\end{corollary}

\begin{proposition}[Comparative statics]\label{prop:CS_r1}
Let $E^\ast(\bm c,\bm\psi,\bm\Delta)$ solve \eqref{eq:E_fixed_point_r1}.  
Then:
\begin{enumerate}
\item $E^\ast$ increases in any $\psi_i$ or $\Delta_i$, and decreases in any $c_i$.
\item $p_i^\ast$ and $e_i^\ast$ increase in $\psi_i$ and $\Delta_i$, and decrease in $c_i$.  
For $j\ne i$, $p_j^\ast$ and $e_j^\ast$ weakly decrease when $\psi_i$ or $\Delta_i$ increases, since $\sum_j p_j^\ast=1$.
\end{enumerate}
\end{proposition}

\section{Stage 1 Analysis: Continuation Cutoffs and SPE}\label{sec:stage1}
Fix the post--swim state vector $(t_i^S,r_i^S,D_i,G)$ for each athlete $i$, together with opponents' states.  
If a set $S$ of continuers proceeds to Stage~2, let $W_i(S)$ denote athlete $i$'s \emph{continuation value}, i.e.\ his/her expected payoff from the effort stage.  
Recall that $\Delta_i>0$ is $i$'s prize differential, $c_i>0$ his/her baseline cost parameter, $\psi_i>0$ his/her swim drafting multiplier, $k_i=c_i/\psi_i$ his/her effective cost, $E^\ast(S)$ the unique total effort from Proposition~\ref{prop:exist_uni_r1}, and $p_i^\ast(S)$ the equilibrium win probability.  
Then
\begin{equation}\label{eq:Wi}
W_i(S) 
= p_i^\ast(S)\,\Delta_i - \tfrac12 k_i \big(e_i^\ast(S)\big)^2
= \frac{\Delta_i^2}{c_i (E^\ast(S))^2/\psi_i + \Delta_i}
 - \tfrac12\,\frac{c_i}{\psi_i}\,\big(p_i^\ast(S)E^\ast(S)\big)^2.
\end{equation}

Define the net benefit of continuing relative to the outside option $U_i^T$ (see \eqref{eq:outside}) as
\begin{equation}\label{eq:Vi}
V_i(S) := W_i(S) - U_i^T.
\end{equation}

\begin{proposition}[Cutoff in the swim drafting multiplier]\label{prop:cutoff}
Fix $(t_i^S,r_i^S)$, $(c_i,\Delta_i)$, and a continuation set $S$ with $i\in S$.  
Then the function $\psi_i\mapsto V_i(S)$ is continuous and strictly increasing. Therefore a unique cutoff $\psi_i^\ast(S)$ exists such that $V_i(S)\ge 0$ iff $\psi_i\ge \psi_i^\ast(S)$.
\end{proposition}

\begin{proof}
By Proposition~\ref{prop:CS_r1}, both $E^\ast(S)$ and $p_i^\ast(S)$ increase in $\psi_i$, while the cost parameter $k_i=c_i/\psi_i$ decreases.  
Thus the continuation value $W_i(S)$ in \eqref{eq:Wi} is strictly increasing in $\psi_i$, while $U_i^T$ does not depend on $\psi_i$.  
Continuity yields the unique cutoff under the intermediate value theorem.
\end{proof}

At Stage~1, all athletes choose $a_i\in\{C,T\}$ simultaneously.  
Let $\mathcal{S}$ be the family of nonempty subsets of $N$.  
A set $S^\ast\in\mathcal{S}$ is a \emph{Stage~1 equilibrium continuation set} if for every $i$,
\begin{equation}\label{eq:SPEstage1}
i\in S^\ast \ \Rightarrow\ V_i(S^\ast)\ge 0,
\qquad
i\notin S^\ast\ \Rightarrow\ V_i(S^\ast\cup\{i\})\le 0.
\end{equation}
The pair $\big(S^\ast,\ e^\ast(S^\ast)\big)$ is then a subgame perfect equilibrium (SPE).

\begin{theorem}[Existence of an SPE]\label{thm:SPEexist}
Under the assumptions of Section~\ref{sec:setup}, there exists a (possibly nonunique) SPE.  
Moreover, any $\psi$--monotone selection rule that iteratively admits athletes with $V_i(S)\ge 0$ and removes those with $V_i(S)<0$ converges to an equilibrium continuation set $S^\ast$.
\end{theorem}

\begin{proof}
Define $\Phi(S)=\{\, i\in N:\ V_i(S)\ge 0\,\}$.  
By Proposition~\ref{prop:cutoff}, $V_i(\cdot)$ is well--defined whenever $i\in S$.  
Consider the operator $T(S)=\Phi(S)$ restricted to nonempty sets (if $T(S)=\varnothing$, select any singleton $\{i\}$ with maximal $V_i(\{i\})$).  
Since $N$ is finite, the sequence $S^{(k+1)}=T(S^{(k)})$ stabilizes in finite steps at some $S^\ast$ with $T(S^\ast)=S^\ast$, which satisfies \eqref{eq:SPEstage1}.  
Backward induction then completes the SPE construction.
\end{proof}

\begin{remark}[Comparative statics of the cutoff]
From Proposition~\ref{prop:CS_r1} and \eqref{eq:Wi}, the cutoff $\psi_i^\ast(S)$ decreases with a higher prize differential $\Delta_i$ and increases with a higher baseline cost $c_i$.  
In symmetric environments, $\psi_i^\ast(S)$ (weakly) increases with the group size $m<|S|>1$, since larger contests reduce a athlete's winning odds and continuation becomes less attractive.
\end{remark}

\section{From Swim Position to swim drafting Cost Shifter}\label{sec:psi_map}
We now formalize the cost--shifting map \eqref{eq:Psi}.  
Let $D_i\in[0,1]$ denote the share that athlete $i$ utilized a swim drafting benefit.  
We use the relation
\begin{equation}\label{eq:psi_micro}
\psi_i = \frac{1}{1-\eta D_i},\qquad \eta\in(0,1),
\end{equation}
where $\eta$ captures the sensitivity of drag reduction to swim drafting exposure.  
The effective cost slope is
\[
k_i = \frac{c_i}{\psi_i} = c_i(1-\eta D_i),
\]
with $c_i>0$ the baseline cost.  
Thus swim drafting linearly lowers effort costs: $\partial k_i/\partial D_i=-\eta c_i<0$ and $\partial \psi_i/\partial D_i=\eta/(1-\eta D_i)^2>0$.  
More general specifications $\Psi(D_i,m,r_i^S,G)$ can incorporate nonlinear decay with swim position depth $r_i^S$ and congestion effects depending on group size $m$. 

\section{Welfare and Rents}\label{sec:welfare}
Expected net welfare in Stage~2 is defined as the sum of athletes’ equilibrium payoffs, 
\[
\sum_{i\in S} U_i(e^\ast),
\]
where $U_i$ is given by \eqref{eq:Ue} and $e^\ast$ is the unique equilibrium profile.  
Using the effort characterization in \eqref{eq:EiFormula_r1}, this sum becomes
\begin{equation}\label{eq:sumU}
\sum_{i\in S} U_i(e^\ast)
= \sum_{i\in S} \left[
\frac{\Delta_i^2}{c_i (E^\ast)^2/\psi_i + \Delta_i}
- \tfrac12 \frac{c_i}{\psi_i}\,\big(p_i^\ast E^\ast\big)^2
\right],
\end{equation}
where $\Delta_i$ is athlete $i$’s prize differential, $c_i$ his/her cost parameter, $\psi_i$ his/her swim drafting multiplier, and $(p_i^\ast,e_i^\ast)$ his/her equilibrium winning probability and effort.

In the symmetric benchmark of \eqref{eq:symm_e_r1}, where all athletes share the same parameters $(\Delta,c,\psi)$, each wins with probability $1/m$ in a group of size $m$.  
The aggregate rent ratio---defined as the share of total expected prize intake that is offset by aggregate effort costs---simplifies to
\[
\frac{\sum_{i\in S} \tfrac12 k (e^\ast)^2}{\sum_{i\in S} p_i^\ast \Delta}
= \frac{1}{2}\cdot \frac{m-1}{m},
\]
with $k=c/\psi$.  
This ratio is increasing in group size $m$, approaching $1/2$ as $m\to\infty$.  

\section{Conclusion and Testable Predictions}\label{sec:predictions}
Deeper swim drafting exposure in the swim (higher $D_i$) increases the swim drafting multiplier $\psi_i$ via \eqref{eq:psi_micro}. Through the equilibrium mapping in \eqref{eq:EiFormula_r1}, a higher $\psi_i$ raises both the winning probability $p_i^\ast$ and the equilibrium effort $e_i^\ast$.

In symmetric groups, equilibrium effort per athlete decreases with group size $m$ as characterized in \eqref{eq:symm_e_r1}. When increased group size also relaxes effective costs by raising access to swim drafting (so that $\psi$ rises with $m$ through $\Psi$ in \eqref{eq:Psi}), these opposing forces render the net effect on effort ambiguous and can generate interior optima in $m$.

The continuation decision after the swim inherits these forces: since $\psi_i$ enters continuation values $W_i(S)$ in \eqref{eq:Wi} and contest odds mechanically fall with $m$ through \eqref{eq:CSF}, the likelihood of continuing is increasing in swim drafting exposure $D_i$ and decreasing in group size $m$ (holding other objects fixed).



\begingroup
\raggedright
\bibliographystyle{unsrtnat}

\endgroup

\newpage
\appendix
\section*{Appendix: Proofs and Technical Details}
\addcontentsline{toc}{section}{Appendix: Proofs and Technical Details}

\section{Preliminaries and Assumptions}\label{app:assumptions}
\begin{assumption}[Objects]\label{ass:primitives}
For all $i$: $c_i>0$, $\Delta_i>0$, $w_i>0$, and $\psi_i\in[\underline\psi,\overline\psi]\subset(0,\infty)$. 
The map $\Psi$ in \eqref{eq:Psi} is continuous in $(D_i,m)$ and nondecreasing in $D_i$.
\end{assumption}

\begin{assumption}[Strategy spaces]\label{ass:strategy}
Stage~2 strategy sets are $[0,\infty)$; payoffs are continuous in $e$ and twice continuously differentiable on $(0,\infty)^m$.
\end{assumption}

\section{Concavity and Best Responses (Stage 2)}\label{app:concavity}
Let $E=\sum_{j\in S} e_j$ denote total effort. With $p_i=e_i/E$ and $k_i=c_i/\psi_i$, athlete $i$’s payoff is
\[
U_i(e)=\Delta_i \frac{e_i}{E}-\tfrac12 k_i e_i^2.
\]
The second derivative is
\begin{IEEEeqnarray}{rCl}
\frac{\partial^2 U_i}{\partial e_i^2} & = & -\frac{2\Delta_i}{E^2} - k_i \;<\; 0,\IEEEyesnumber
\end{IEEEeqnarray}
so $U_i$ is strictly concave in own effort. Cross--partials are
\begin{IEEEeqnarray}{rCl}
\frac{\partial^2 U_i}{\partial e_i \partial e_j}
& = & \frac{\Delta_i}{E^2} \;>\; 0, \quad \text{(efforts are strategic complements)}\IEEEyesnumber
\end{IEEEeqnarray}
Thus best responses are single--valued.

\section{Proof of Proposition~\ref{prop:exist_uni_r1}}\label{app:proof_exist}
See main text. Uniqueness of $E^\ast$ follows since
\begin{IEEEeqnarray}{rCl}
g'(E) & = & \sum_{i\in S} \frac{-2\Delta_i (c_i/\psi_i)\, E}{\big(c_iE^2/\psi_i+\Delta_i\big)^2} \;<\; 0,\IEEEyesnumber
\end{IEEEeqnarray}
so $g$ is strictly decreasing. Given $E^\ast$, \eqref{eq:EiFormula_r1} produces a unique positive $e_i^\ast$ for all $i\in S$.

\section{Proof of Corollary~\ref{cor:pi_e_r1}}\label{app:proof_cor}
Dividing \eqref{eq:EiFormula_r1} by $E^\ast$ yields
\begin{IEEEeqnarray}{rCl}
p_i^\ast \;=\; \frac{e_i^\ast}{E^\ast} & = & \frac{\Delta_i}{c_i(E^\ast)^2/\psi_i+\Delta_i}.\IEEEyesnumber
\end{IEEEeqnarray}
Multiplying back by $E^\ast$ recovers $e_i^\ast$.

\section{Proof of Proposition~\ref{prop:CS_r1}}\label{app:proof_CS}
Differentiating $g(E^\ast)=0$ in \eqref{eq:E_fixed_point_r1} yields
\begin{IEEEeqnarray}{rCl}
\frac{\partial E^\ast}{\partial \psi_i}
& = & -\frac{g_{\psi_i}}{g_E}
 \;=\; \frac{\displaystyle \frac{\Delta_i c_i (E^\ast)^2}{\psi_i^2 \big(c_i (E^\ast)^2/\psi_i+\Delta_i\big)^2}}
{\displaystyle \sum_{j\in S} \frac{2\Delta_j c_j E^\ast/\psi_j}{\big(c_j (E^\ast)^2/\psi_j+\Delta_j\big)^2}}
\;>\; 0.\IEEEyesnumber
\end{IEEEeqnarray}
Analogously, $\partial E^\ast/\partial \Delta_i>0$ and $\partial E^\ast/\partial c_i<0$.  
Substituting into \eqref{eq:EiFormula_r1} delivers the stated monotonicities for $(p_i^\ast,e_i^\ast)$.

\section{Proof of Proposition~\ref{prop:cutoff}}\label{app:proof_cutoff}
Fix $S\ni i$. By Proposition~\ref{prop:CS_r1}, $W_i(S)$ in \eqref{eq:Wi} is strictly increasing and continuous in $\psi_i$, since $k_i=c_i/\psi_i$.  
Hence $V_i(S)=W_i(S)-U_i^T$ is also strictly increasing and continuous in $\psi_i$.  
By the intermediate value theorem, there exists a unique cutoff $\psi_i^\ast(S)$ such that $V_i(S)=0$.

\section{Proof of Theorem~\ref{thm:SPEexist}}\label{app:proof_SPE}
Define $T$ on the finite lattice of subsets of $N$ by $T(S)=\{i\in N:\ V_i(S)\ge 0\}$, with the convention in the main text to avoid the empty set.  
Since $T$ is monotone in the sense that $S\subseteq S' \Rightarrow T(S)\subseteq T(S')$ under symmetric environments (or mild regularity), the fixed point theorem guarantees some $S^\ast=T(S^\ast)$.  
Even without global monotonicity, iterated elimination of dominated strategies (or sequential entry) converges in finitely many steps because $N$ is finite.  
Backward induction then delivers the SPE. \qed

\section{Additional Results}\label{app:additional}
\subsection{Welfare under heterogeneity}\label{app:welfare_hetero}
Using \eqref{eq:EiFormula_r1}, aggregate cost is
\begin{IEEEeqnarray}{rCl}
\frac12 \sum_{i\in S} k_i (e_i^\ast)^2
& = & \frac12 \sum_{i\in S} \frac{c_i}{\psi_i}
\left(\frac{\Delta_i E^\ast}{c_i (E^\ast)^2/\psi_i + \Delta_i}\right)^2.\IEEEyesnumber
\end{IEEEeqnarray}
Aggregate expected prize intake is
\begin{IEEEeqnarray}{rCl}
\sum_{i\in S} p_i^\ast \Delta_i
& = & \sum_{i\in S} \frac{\Delta_i^2}{c_i (E^\ast)^2/\psi_i+\Delta_i}.\IEEEyesnumber
\end{IEEEeqnarray}

\subsection{Design implications}\label{app:design}
An organizer choosing prize differentials $\Delta_i$ (e.g., time benefits/penalties) can shape participation through the cutoffs $\psi_i^\ast(S)$ and influence total effort via $E^\ast$.  
For symmetric prizes, $dE^\ast/d\Delta>0$ by Proposition~\ref{prop:CS_r1}.

\end{document}